\pdfoutput=1
\RequirePackage{ifpdf}
\ifpdf 
\documentclass[pdftex]{sigma}
\else
\documentclass{sigma}
\fi

\numberwithin{equation}{section}

\newtheorem{Theorem}{Theorem}[section]
\newtheorem{Corollary}[Theorem]{Corollary}
\newtheorem{Proposition}[Theorem]{Proposition}
 { \theoremstyle{definition}
\newtheorem{Example}[Theorem]{Example}
\newtheorem{Remark}[Theorem]{Remark} }

\newcommand\la{{\lambda}}
\newcommand\al{{\alpha}}
\newcommand\be{{\beta}}
\newcommand\rd{{\mathrm{d}}}
\newcommand\ri{{\mathrm{i}}}
\newcommand{\Q}{{\mathbb Q}}
\newcommand{\Z}{{\mathbb Z}}
\newcommand{\C}{{\mathbb C}}

\begin{document}


\newcommand{\arXivNumber}{1705.01094}

\renewcommand{\thefootnote}{}

\renewcommand{\PaperNumber}{057}

\FirstPageHeading

\ShortArticleName{On Reductions of the Hirota--Miwa Equation}

\ArticleName{On Reductions of the Hirota--Miwa Equation\footnote{This paper is a~contribution to the Special Issue on Symmetries and Integrability of Dif\/ference Equations. The full collection is available at \href{http://www.emis.de/journals/SIGMA/SIDE12.html}{http://www.emis.de/journals/SIGMA/SIDE12.html}}}

\Author{Andrew N.W.~HONE, Theodoros E.~KOULOUKAS and Chloe WARD}
\AuthorNameForHeading{A.N.W.~Hone, T.E.~Kouloukas and C.~Ward}
\Address{School of Mathematics, Statistics and Actuarial Science, University of Kent,\\ Canterbury CT2 7NF, UK}
\Email{\href{mailto:A.N.W.Hone@kent.ac.uk}{A.N.W.Hone@kent.ac.uk}, \href{mailto:T.E.Kouloukas@kent.ac.uk}{T.E.Kouloukas@kent.ac.uk}, \href{mailto:chloe_ward@live.co.uk}{chloe$\_$ward@live.co.uk}}
\URLaddress{\url{https://www.kent.ac.uk/smsas/our-people/profiles/hone_andrew.html}}

\ArticleDates{Received May 02, 2017, in f\/inal form July 17, 2017; Published online July 23, 2017}

\Abstract{The Hirota--Miwa equation (also known as the discrete KP equation, or the octahedron recurrence) is a~bilinear partial dif\/ference equation in three independent variables. It is integrable in the sense that it arises as the compatibility condition of a linear system (Lax pair). The Hirota--Miwa equation has inf\/initely many reductions of plane wave type (including a quadratic exponential gauge transformation), def\/ined by a triple of integers or half-integers, which produce bilinear ordinary dif\/ference equations of Somos/Gale--Robinson type. Here it is explained how to obtain Lax pairs and presymplectic structures for these reductions, in order to demonstrate Liouville integrability of some associated maps, certain of which are related to reductions of discrete Toda and discrete KdV equations.}


\Classification{70H06; 37K10; 39A20; 39A14; 13F60}

\renewcommand{\thefootnote}{\arabic{footnote}}
\setcounter{footnote}{0}

\section{Introduction}
The Hirota--Miwa equation, or discrete Kadomtsev--Petviashvili equation (discrete KP), is the bilinear partial dif\/ference equation
\begin{gather} \label{Hirota}
T_1T_{-1}=T_{2}T_{-2}+T_3T_{-3},
\end{gather}
which serves as a generating equation for the full KP hierarchy of partial dif\/ferential equations~\cite{dHirota,Miwa}. In the above,
the tau function $T$ is a function of three independent variables, $T=T(m_1,m_2,m_3)$, and subscripts denote shifts, i.e.\ $T_{\pm i}= T|_{m_i\rightarrow m_i \pm 1}$, for $i=1,2,3$. The equation (\ref{Hirota}) is integrable in the sense that it arises as the compatibility condition of the following linear system (Lax pair) \cite{Krich}:
\begin{gather}
T_{-1,3} \Psi_{1,2}+T \Psi_{2,3}= T_{2,3} \Psi, \nonumber\\
T \Psi_{-1,2}+T_{-1,3} \Psi_{2,-3}=T_{-1,2} \Psi\label{linHir}
\end{gather}
(for a comparison of alternative Lax pairs see~\cite{zabrodin}). From another viewpoint of discrete integrability, the equation (\ref{Hirota}) can be extended to a system with an arbitrary number of independent variables, in which case it satisf\/ies the multidimensional consistency property, and it can be further extended to a system with self-consistent sources (see~\cite{doliwalin} and references).

In this paper we consider reductions of (\ref{Hirota}) to bilinear ordinary dif\/ference equations, in order to show that they can be interpreted as discrete integrable systems in an appropriate sense, namely that Liouville's theorem for symplectic maps can be applied \cite{maeda, veselov}. Certainly it should be no surprise that reductions of a discrete integrable system to a lower number of independent variables produce integrable maps. Indeed, from a dif\/ferent point of view, one can consider the algebraic entropy of such reductions, and show that all one-dimensional reductions of (\ref{Hirota}) have quadratic degree growth \cite[Theorem~6.8]{mase}, so their algebraic entropy is zero. However, in general the connections between algebraic entropy, Lax integrability and Liouville integrability are rather subtle, and the bilinear ordinary dif\/ference equations that arise as plane wave reductions of (\ref{Hirota}) are not symplectic as they stand. In order to interpret them correctly, we make use of the properties of recurrences derived from cluster algebras \cite{FM}, which provide an appropriate presymplectic structure for bilinear dif\/ference equations \cite{FH} (see also \cite{fg,gsv,inouen} and references), and a further reduction procedure leads to symplectic maps of a certain kind, referred to in \cite{honeinoue} as U-systems. By applying the reduction process to the linear equations~(\ref{linHir}), we f\/ind a Lax pair for the corresponding bilinear ordinary dif\/ference equations, leading to the construction of conserved quantities, and then in each case the problem remains to check that this provides a suf\/f\/icient number of f\/irst integrals in involution for the associated U-system.

General properties of plane wave reductions of the Hirota--Miwa equation are described in Section~\ref{section2}, and one particular example is considered in full detail. This example turns out to be a member of a~two-parameter family of bilinear equations that arise from reductions of a lattice equation of discrete Toda type, considered in Section~\ref{section3}. In Section~\ref{section4} we consider another two-parameter family, consisting of pairs of bilinear equations related to discrete KdV reductions, and we consider another example in detail, before making some conclusions in Section~\ref{section5}.

\section{Plane wave reductions}\label{section2}

The Hirota--Miwa equation (\ref{Hirota}) admits the plane wave reduction
\begin{gather} \label{taf}
T(m_1,m_2,m_3)=a_1^{{m_1}^2}a_2^{{m_2}^2}a_3^{{m_3}^2}\tau_m, \qquad m=m_0+ \delta_1 m_1+\delta_2 m_2+\delta_3 m_3,
\end{gather}
where $m_0$ is arbitrary and the distinct values of the parameters $\delta_1$, $\delta_2$, $\delta_3$ are all chosen to be integers or half-integers (the case when any of the $\delta_j$ coincide is not interesting). The above expression includes a quadratic exponential gauge transformation, which allows the insertion of coef\/f\/icients into the right-hand side of the parameter-free equation (\ref{Hirota}). By substituting (\ref{taf}) into (\ref{Hirota}),
it follows that $\tau_m$ satisf\/ies the autonomous bilinear ordinary dif\/ference equation
\begin{gather*} 
\tau_{m+\delta_1}\tau_{m-\delta_1}=\alpha \tau_{m+\delta_2}\tau_{m-\delta_2}+ \beta \tau_{m+\delta_3}\tau_{m-\delta_3},
\end{gather*}
where the ratios of the arbitrary parameters $a_j$ in (\ref{taf}) yield the coef\/f\/icients
\begin{gather*}\alpha=\frac{a_2^2}{a_1^2}, \qquad \beta=\frac{a_3^2}{a_1^2}.\end{gather*}
We can shift the last equation by $\delta_1$ and rewrite it as
\begin{gather} \label{RedHirota2}
\tau_{m+N}\tau_{m}=\alpha \tau_{m+\delta_1+\delta_2}\tau_{m+\delta_1-\delta_2}+ \beta \tau_{m+\delta_1+\delta_3}\tau_{m+\delta_1-\delta_3} \qquad \mathrm{with} \quad N=2\delta_1,
\end{gather}
and without loss of generality we can assume that
\begin{gather*} \delta_1> \max (\delta_2,\delta_3), \end{gather*}
so that the overall order of the recurrence (\ref{RedHirota2}) is $N$.

The iteration of the recurrence (\ref{RedHirota2}) is equivalent to iteration of the birational map $\varphi\colon$ $\mathbb{C}^N \rightarrow \mathbb{C}^N$, given by
\begin{gather} \label{mapphi}
\varphi(\tau_0,\tau_1, \ldots,\tau_{N-2},\tau_{N-1})= \left(\tau_1,\tau_2, \ldots,\tau_{N-1},\frac{\alpha \tau_{\delta_1+\delta_2}\tau_{\delta_1-\delta_2}+\beta \tau_{\delta_1+\delta_3} \tau_{\delta_1-\delta_3}}{\tau_0}\right),
\end{gather}
which is an example of a cluster map: it arises from a sequence of mutations in a cluster algebra, constructed from a quiver that is cluster mutation-periodic with period~1, with the coef\/f\/i\-cients~$\alpha$,~$\beta$ corresponding to frozen variables \cite{FM}. By making gauge transformations one can also consider the original equation (\ref{Hirota}) with non-autonomous coef\/f\/icients, and the same is true for its reductions (see~\cite{mase}). In due course we will also f\/ind it useful to allow the coef\/f\/icients~$\alpha$,~$\beta$ to depend on the index~$m$.

\subsection{Reduction of Lax pairs}

We can apply the plane wave reduction (\ref{taf}) at the level of the Lax pair. However, before doing so it is important to note that the compatibility condition of the linear system (\ref{linHir}) is not exactly~(\ref{Hirota}), but rather
\begin{gather} \label{rform}
R_{1,-3}=R, \qquad R=\frac{T_1T_{-1}-T_3T_{-3}}{T_2T_{-2}}.
\end{gather}
In order to get precisely (\ref{Hirota}), corresponding to $R\equiv 1$, it is necessary to augment the pair~(\ref{linHir}) with a third equation, to get a Lax triad. However, for our purposes it will be convenient to leave this slight ambiguity.

To make the reduction, we require that the tau function $T$ is given by (\ref{taf}), and take the wave function~$\Psi$ in the form
\begin{gather} \label{subs2}
\Psi(m_1,m_2,m_3)=\lambda_1^{m_1} \lambda_2^{m_2} \lambda_3^{m_3}T(m_1,m_2,m_3) \phi_m,
\end{gather}
where $\la_j$ correspond to spectral parameters, and upon setting
 \begin{gather*}\zeta=\lambda_2\lambda_1^{-1}, \qquad \xi=\big(a_1^2\lambda_1\lambda_2\big)^{-1}, \qquad \lambda_3=a_3^2 \lambda_2\end{gather*}
we can derive a corresponding linear system for the reductions (\ref{RedHirota2}) of (\ref{Hirota}).

\begin{Proposition} \label{propLS}
The plane wave reduction of the linear system \eqref{linHir} is
\begin{gather}
Y_m  \phi_{m+\delta_1+\delta_2}+\beta\zeta \phi_{m+\delta_2+\delta_3} =  \xi \phi_m, \nonumber\\
\phi_{m+\delta_1-\delta_2}-X_m \phi_{m+\delta_1-\delta_3} =  \zeta \phi_m,\label{LaxRe1}
\end{gather}
where
\begin{gather*}
X_m=\frac{\tau_{m+2\delta_1-2\delta_3}\tau_{m+\delta_1-\delta_2}}{\tau_{m+2\delta_1-\delta_2-\delta_3}\tau_{m+\delta_1-\delta_3}}, \qquad
 Y_m=\frac{\tau_{m+2\delta_1+\delta_2-\delta_3}\tau_{m}}{\tau_{m+\delta_1+\delta_2}\tau_{m+\delta_1-\delta_3}}.
\end{gather*}
The compatibility condition of \eqref{LaxRe1} is a Somos-$N$ recurrence with a periodic coefficient, given by
\begin{gather}\label{somosn}
\tau_{m+N}\tau_{m}=\alpha_m \tau_{m+\delta_1+\delta_2}\tau_{m+\delta_1-\delta_2}+ \beta \tau_{m+\delta_1+\delta_3}\tau_{m+\delta_1-\delta_3},
\qquad \alpha_{m+\delta_1-\delta_3}=\alpha_m,
\end{gather}
with $N=2\delta_1$.
\end{Proposition}

\begin{proof} This follows by substituting (\ref{taf}) and (\ref{subs2}) into (\ref{linHir}): the f\/irst linear equation produces the top equation in~(\ref{LaxRe1}), and the second equation in~(\ref{linHir}) almost immediately yields the bottom equation in~(\ref{LaxRe1}) after shifting the indices on each $\phi_j$ to $\phi_{j+\delta_1-\delta_2}$, except that we have also made an overall shift $m\to m+\delta_1-\delta_3$ in the dependent variables which appear in the combina\-tions~$X_m$,~$Y_m$ (which we are free to do, since the compatibility of the linear system gives autonomous relations between these quantities). Now observe that the parameter~$\alpha$ in~(\ref{RedHirota2}) does not appear anywhere in the linear system~(\ref{LaxRe1}), and the compatibility condition of the latter is obtained by substituting~(\ref{taf}) into~(\ref{rform}), which implies that $\tau_m$ satisf\/ies the non-autonomous bilinear recurrence~(\ref{somosn}) with the coef\/f\/icient~$\alpha_m$ of period $\delta_1-\delta_3$.
\end{proof}

Each iteration of (\ref{somosn}) is equivalent to an iteration of a map $\varphi_m\colon \C^N\to \C^N$, which is of the same form as (\ref{mapphi}), except that the coef\/f\/icient $\al$ is replaced by $\alpha_m$, depending on $m\bmod \delta_1-\delta_3$. The iterates of this map have a matrix Lax representation, and the dynamics preserves an associated spectral curve.

\begin{Corollary} \label{LaxPair}
The scalar Lax pair \eqref{LaxRe1} is equivalent to a~matrix linear system of size $K=\max (\delta_1-\delta_2, \delta_1-\delta_3)$, of the form
\begin{gather}\label{mlaxp}
 {\bf L}_m(\zeta)  \mathbf{\Phi}_m=\xi \mathbf{\Phi}_m, \qquad \mathbf{\Phi}_{m+1}= {\bf M}_m(\zeta)\mathbf{\Phi}_m .
\end{gather}
The compatibility condition of the latter system is the discrete Lax equation
\begin{gather}\label{dlax}
{\bf L}_{m+1}{\bf M}_m= {\bf M}_m{\bf L}_m,
\end{gather}
which preserves the spectral curve in the $(\zeta,\xi)$ plane given by
\begin{gather}\label{specc}
{\cal P}(\zeta,\xi)\equiv \det ({\bf L}_m(\zeta) -\xi \mathbf{1}) = 0.
\end{gather}
\end{Corollary}
\begin{proof} Upon introducing the vector $\mathbf{\Phi}_m=(\phi_m,\phi_{m+1},\ldots, \phi_{m+K-1})^T$, the top equation in (\ref{LaxRe1}) becomes an eigenvalue equation for a matrix ${\bf L}_m$, while the bottom equation can be rewritten in matrix form with another matrix~${\bf M}_m$, where ${\bf L}_m$, ${\bf M}_m$ in~(\ref{mlaxp}) are $K\times K$ matrix functions of $\zeta$ and the dynamical variables. The compatibility condition~(\ref{dlax}) shows that this is an isospectral evolution, preserving the coef\/f\/icients of the polynomial~${\cal P}(\zeta,\xi)$.
\end{proof}

\begin{Example}\label{todaex}
Choosing $(\delta_1,\delta_2,\delta_3)=(4,3,0)$ in (\ref{somosn}) gives rise to the recurrence
\begin{gather} \label{874}
\tau_{m+8}\tau_{m}=\alpha_m \tau_{m+7}\tau_{m+1}+ \beta \tau_{m+4}^2, \qquad \alpha_{m+4}=\alpha_m.
\end{gather}
As we shall see below, the latter recurrence also arises as a reduction of a~5-point lattice equation of discrete Toda type appearing in~\cite{book}. In this case, the scalar linear system (\ref{LaxRe1}) consists of two equations, of orders 7 and 4 respectively, with coef\/f\/icients
\begin{gather*}
X_m=\frac{\tau_{m+8}\tau_{m+1}}{\tau_{m+5}\tau_{m+4}}, \qquad  Y_m=\frac{\tau_{m+11}\tau_{m}}{\tau_{m+7}\tau_{m+4}}.
\end{gather*}
After using the equation of fourth order to eliminate higher shifts of $\phi_m$ from the equation of order 7, this gives the system
\begin{gather*}
\phi_{m+4} = \frac{1}{X_m}(\phi_{m+1}-\zeta\phi_m), \\
\xi \phi_m  =  \zeta\left(\beta - \frac{Y_m}{X_{m+3}}\right)\phi_{m+3}+\frac{Y_m}{X_mX_{m+3}}(\phi_{m+1}-\zeta\phi_m).
\end{gather*}
The matrix entries in the Lax pair of Corollary \ref{LaxPair} can then be written in a fairly compact form in terms of the quantities $X_{m+j}$ for $0\leq j\leq 6$, $Y_{m+k}$ for $0\leq k\leq 3$ and $\beta$, $\zeta$, in which case the compatibility conditions are
\begin{gather}\label{co1}
X_mY_{m+4}= X_{m+7}Y_{m+1}, \qquad X_mY_{m+4}- X_{m+7}Y_{m}=\beta X_{m+7}(X_m-X_{m+3}).
\end{gather}
When these equations are rewritten in terms of the tau function $\tau_m$, the f\/irst of them is just a~tautology, while the second one says that
\begin{gather*}
\frac{\tau_{m+12}\tau_{m+4}-\beta\tau_{m+8}^2}{\tau_{m+11}\tau_{m+5}}= \frac{\tau_{m+8}\tau_{m}-\beta\tau_{m+4}^2}{\tau_{m+7}\tau_{m+1}},
\end{gather*}
which is equivalent to (\ref{874}). We can further use the equation (\ref{874}) to rewrite the entries of~${\bf L}_m$,~${\bf M}_m$ in terms of
the tau function and~$\alpha_m$ as well as~$\beta$, $\zeta$, but rather than doing this here we follow~\cite{Chloe} and introduce the quantity
\begin{gather}\label{zdef}
z_m =\frac{\tau_{m+5}\tau_m}{\tau_{m+4}\tau_{m+1}},
\end{gather}
so that $X_m=z_{m+1}z_{m+2}z_{m+3}$, $ Y_m=z_m X_{m}X_{m+3}$, to f\/ind (setting $m\to 0$ for convenience)
\begin{gather}
\mathbf{L} =
\begin{pmatrix}
 -\zeta z_0 & z_0 & 0 & \zeta(\beta-z_0 z_{1} z_{2} z_{3}) \\
\!\zeta^2\left(z_{4}-\frac{ \beta}{z_{1} z_{2} z_{3}}\right) & \!\zeta\left(\frac{\beta}{z_{1} z_{2} z_{3}}- z_{1} - z_{4} \right) & z_{1} & 0 \\
0 &\zeta^2 \left( z_{5}-\frac{\beta}{z_{2} z_{3} z_{4}} \right) & \!\zeta \left( \frac{ \beta}{z_{2} z_{3} z_{4}}-z_{2}-z_{5}\right) & z_{2} \\
-\frac{\zeta}{z_{1} z_{2}} &\frac{1}{z_{1} z_{2}} & \zeta^2 \left(z_{6}-\frac{\beta}{z_{3} z_{4} z_{5}}\right) & \!\zeta\left(\frac{ \beta}{z_{3} z_{4} z_{5}} -z_{3}- z_{6}\right)\!
\end{pmatrix},\nonumber
\\ \label{tolax}
\mathbf{M} =
\begin{pmatrix}
0 & 1 & 0 & 0 \\
0 & 0 & 1 & 0 \\
0 & 0 & 0 & 1 \\
\frac{-\zeta}{z_1 z_2 z_3} & \frac{1}{z_1 z_2 z_3} & 0 & 0
\end{pmatrix},
\end{gather}
with ${\bf L}={\bf L}_0$, ${\bf M}={\bf M}_0$, in which case the second equation in (\ref{co1}) becomes a recurrence of order~7 for~$z_m$, namely
\begin{gather}\label{co2}
z_{m+7}-z_m =\beta\left(\frac{1}{z_{m+4}z_{m+5}z_{m+6}}- \frac{1}{z_{m+1}z_{m+2}z_{m+3}}\right).
\end{gather}

The corresponding spectral curve (\ref{specc}) is of genus 9, and takes the form
\begin{gather}\label{todaspec}
\Gamma\colon \  \xi^4+H_1\zeta\xi^3+H_2\zeta^2\xi^2+\big(H_3\zeta^3-1\big)\xi+H_4\zeta^7 +\beta\zeta^4=0,
\end{gather}
where
\begin{gather}
H_1 =\sum_{k=0}^6 z_k -\beta\sum_{k=0}^2 \frac{1}{z_{k+1}z_{k+2}z_{k+3}}, \nonumber\\
H_2 =  \sum_{k=0}^6 z_kz_{k+1}+z_kz_{k+2}\nonumber\\
\hphantom{H_2 =}{} -
\beta\left(\frac{1}{z_1z_3} +\frac{1}{z_2z_3} +\frac{1}{z_2z_4} +\frac{1}{z_3z_4} +\frac{1}{z_3z_5} +\frac{z_0+z_1}{z_3z_4z_5} +\frac{z_0+z_6}{z_2z_3z_4} +\frac{z_5+z_6}{z_1z_2z_3}\right)\nonumber\\
\hphantom{H_2 =}{}+ \beta^2\left(\frac{1}{z_1z_2^2z_3^2z_4} +\frac{1}{z_1z_2z_3^2z_4z_5} +\frac{1}{z_2z_3^2z_4^2z_5} \right) ,\nonumber\\
H_3= \sum_{k=0}^6z_kz_{k+1}z_{k+2} -\beta\left(\frac{1}{z_3} +\frac{z_0}{z_3z_4} +\frac{z_6}{z_2z_3} +\frac{z_0z_1}{z_3z_4z_5} +\frac{z_0z_6}{z_2z_3z_4} +\frac{z_5z_6}{z_1z_2z_3} \right)\nonumber\\
\hphantom{H_3=}{}  +\beta^2\left(\frac{1}{z_1z_2z_3^2z_4} +\frac{1}{z_2z_3^2z_4z_5} +\frac{z_0}{z_2z_3^2z_4^2z_5} +\frac{z_6}{z_1z_2^2z_3^2z_4} \right) -\frac{\beta^3}{z_1z_2^2z_3^3z_4^2z_5} ,\nonumber\\
H_4 = \frac{1}{z_3}\prod_{k=0}^3\left(z_kz_{k+3}-\frac{\beta}{z_{k+1}z_{k+2}}\right),\label{h1}
\end{gather}
with indices read $\bmod$ 6 where necessary. These coef\/f\/icients of the spectral curve provide 4 functionally independent integrals for the map in~7 dimensions def\/ined by~(\ref{co2}). We shall consider the question of Liouville integrability when we return to this example in Section~\ref{sympl} below.
\end{Example}

\begin{Remark} \label{ab}
The replacement $(\delta_1,\delta_2,\delta_3)\to(\delta_1,\delta_3,\delta_2)$ leads to a recurrence with the same terms as
(\ref{somosn}), but with the roles of the coef\/f\/icients $\alpha,\beta$ reversed. For instance, in the previous example, choosing the parameters
$(\delta_1,\delta_2,\delta_3)=(4,0,3)$ gives $\delta_1-\delta_3=1$, so with both $\alpha$ and $\beta$ as constant coef\/f\/icients,
 the recurrence is just
\begin{gather*}
\tau_{m+8}\tau_{m}=\alpha \tau_{m+7}\tau_{m+1}+ \beta \tau_{m+4}^2.
\end{gather*}
\end{Remark}

Another method for f\/inding f\/irst integrals of discrete KP reductions was given in \cite{mq}, based on reduction of conservation laws.

\subsection{Somos recurrences and cluster maps}
The plane wave reductions of the discrete KP equation (\ref{Hirota}) are recurrence relations of Somos type. They are also particular examples of cluster maps, which arise from cluster mutations of quivers which are periodic with period~1~\cite{FM}, and this provides an appropriate presymplectic structure (i.e.\ a~closed 2-form of constant rank) that is preserved by the dynamics.

The general Somos-$N$ recurrence is a quadratic recurrence relation of the form
\begin{gather} \label{somosN}
x_{m+N}x_m=\sum_{j=1}^{\left\lfloor{\frac{N}{2}}\right\rfloor} a_j  x_{m+N-j}x_{m+j},
\end{gather}
where $a_j$ are coef\/f\/icients. The case when there are only two non-zero coef\/f\/icients $a_j$, $a_k$ on the right-hand side corresponds to the reductions \eqref{RedHirota2} of the discrete KP equation (\ref{Hirota}), while the case of three non-zero coef\/f\/icients $a_j$, $a_k$, $a_\ell$ with $j+k+\ell=0$ $(\bmod\, N)$ arises from reductions of the discrete BKP equation. These particular cases are also sometimes referred to as 3-term or 4-term Gale--Robinson recurrences \cite{jmz}, respectively (where the total number of terms in the equation is counted), and are the only cases which display the Laurent phenomenon \cite{fz}. In general, if there are more than three terms on the right-hand side of~(\ref{somosN}) then the recurrence does not appear to be integrable: the growth of degrees of the terms, or the growth of logarithmic heights of generic iterates in~$\Q$, is exponential. However, for certain choices of coef\/f\/icients and initial data such recurrences with more terms can still be produced by identities for abelian functions~\cite{beh}.

On the other hand, (with the inclusion of coef\/f\/icients $\alpha,\beta$) the reductions~\eqref{RedHirota2} of the discrete KP equation are particular cases of recurrences of the form
\begin{gather} \label{quivrec}
x_{m+N}x_m=\prod_{j=1}^{N-1} x_{m+j}^{[b_{1,j+1}]_{+}} + \prod_{j=1}^{N-1} x_{m+j}^{[-b_{1,j+1}]_{+}},
\end{gather}
which arise from sequences of mutations in a coef\/f\/icient-free cluster algebra. In the above, $[b]_{+}=\max (b,0)$, and the exponents appearing on the right-hand side belong to the f\/irst row of the exchange matrix $B=(b_{ij})$, an $N \times N$ skew-symmetric integer matrix which def\/ines a~quiver (directed graph) consisting of $N$ nodes without 1- or 2-cycles: the rule is that $[b_{ij}]_+$ is the number of arrows from node~$i$ to node~$j$. It was shown by Fordy and Marsh \cite{FM} that the quiver has cluster mutation-periodicity with period~1, meaning that a~single mutation of the quiver is equivalent to a~cyclic permutation of the nodes, if\/f the matrix entries of~$B$ satisfy
\begin{gather}
b_{j,N} = b_{1,j+1}, \qquad j=1,\ldots,N-1, \nonumber\\
b_{j+1,k+1} = b_{j,k}+b_{1,j+1}[-b_{1,k+1}]_{+}-b_{1,k+1}[-b_{1,j+1}]_+, \qquad 1\le j,k \le N-1.\label{percon1}
\end{gather}
In this case, the above conditions mean that matrix $B$ is completely determined by the elements of its f\/irst row, which are the exponents appearing in (\ref{quivrec}). The cluster map def\/ined by (\ref{quivrec}) is the birational map
\begin{gather}
\begin{array}{@{}lrcl}
\varphi\colon &  \C^N &\rightarrow&   \C^N, \\
&    (x_1,\ldots,x_{N-1},x_{N}) &  \mapsto&  \displaystyle \bigg(x_2,\ldots,x_{N} ,
x_1^{-1}\bigg(\prod_j x_{j+1}^{[b_{1,j+1}]_{+}} + \prod_j x_{j+1}^{[-b_{1,j+1}]_{+}}\bigg)\bigg).
\end{array}\label{cluster}
\end{gather}
Furthermore, the conditions (\ref{percon1}) are also necessary and suf\/f\/icient for the 2-form
\begin{gather} \label{presymp}
\omega=\sum_{i<j}\frac{b_{ij}}{x_i x_j} \rd x_i \wedge \rd x_j,
\end{gather}
to be an invariant presymplectic form for the map, i.e.\ $\varphi^*\omega=\omega$~\cite{FH}. The form~(\ref{presymp}) is log-canonical, i.e.\ it is constant in the coordinates~$\log x_i$.

For the conditions (\ref{percon1}) to hold, the non-zero entries in the f\/irst row of $B$ must be palindromic. Hence, in the case of the discrete KP reductions, with a sum of two quadratic terms on the right-hand side, the pattern of non-zero terms is either of the form $1,\ldots, -2,\ldots ,1$ or $1,\ldots,-1,\ldots,-1,\ldots,1$ (up to an overall choice of sign), and the other entries of~$B$ are completely f\/ixed by this pattern.

\begin{Example} \label{b}
Up to sending $B\to -B$, the exchange matrix for the recurrence in Example~\ref{todaex} is
\begin{gather*}B=\left(\begin{matrix}
0 & 1 & 0 & 0 & -2 & 0 & 0 & 1 \\
-1 & 0 & 1 & 0 & 2 & -2 & 0 & 0 \\
0 & -1 & 0 & 1 & 0 & 2& -2 & 0 \\
0 & 0 & -1 & 0 & 1 & 0 & 2 & -2 \\
2 & -2 & 0 & -1 & 0 & 1 & 0 & 0 \\
0 & 2 & -2 & 0 & -1 & 0 & 1 & 0 \\
0 & 0 & 2 & -2 & 0 & -1 & 0 & 1 \\
-1 & 0 & 0 & 2 & 0 & 0 & -1 & 0
\end{matrix} \right). \end{gather*}
\end{Example}

\subsection{Reduction to symplectic maps}\label{sympl}

In order to investigate Liouville integrability in the context of cluster maps, it is convenient to make a reduction to a symplectic map on a space of even dimension $r=\operatorname{rank}B$, that is the space of leaves of the null foliation for $\omega$. This is achieved by choosing a~$\Z$-basis ${\bf v}_1,\ldots,{\bf v}_r$ for $\operatorname{im}B\cap\Z^N$, and then mapping to a corresponding set of Laurent monomials ${\bf u}=(u_1,\ldots,u_r)$ in the coordinates ${\bf x}=(x_j)$, via
\begin{gather}
\begin{array}{@{}lrcl}
 \pi \colon & \C^N &\rightarrow & \C^r, \\
&  {\bf x} & \mapsto & {\bf u} =\big({\bf x}^{{\bf v}_1},\ldots,{\bf x}^{{\bf v}_r}\big).
\end{array} \label{basis}
\end{gather}
Theorem~2.6 in~\cite{FH} says that this produces a symplectic birational map~$\hat\varphi$ in the reduced coordinates~${\bf u}$. By a further
ref\/inement of this result \cite[Proposition 3.9]{honeinoue} one can choose a special $\Z$-basis consisting of palindromic vectors, such that
$\hat\varphi$ is equivalent to iteration of a~single recurrence relation called the U-system.

\begin{Theorem}\label{FH} For any cluster map $\varphi$ given by \eqref{cluster} with associated exchange matrix~$B$ satisfying~\eqref{percon1}, there is a~palindromic $\Z$-basis for $\operatorname{im}B\cap\Z^N$ $($unique up to an overall sign$)$, such that under the reduction~\eqref{basis}, the map $\hat\varphi$ with $\hat{\varphi} \circ \pi=\pi \circ \varphi$ is equivalent to the iteration of the corresponding U-system, of the form
\begin{gather*}
u_{m+r}u_m=\mathcal{F}(u_{m+1},\ldots,u_{m+r-1}),
\end{gather*}
for a certain rational function $\cal F$. Moreover, $\hat\varphi$ preserves the symplectic form $\hat\omega$ such that \smash{$\pi^{*} \hat{\omega}=\omega$}, which is log-canonical in the coordinates ${\bf u}=(u_1,\ldots,u_r)$.
\end{Theorem}

The above theorem immediately applies to the plane wave reductions~\eqref{RedHirota2} of discrete KP, with coordinates $x_j\to \tau_j$, and the symplectic structure obtained by reduction is unaltered even when non-autonomous coef\/f\/icients are included, in particular when $\al$ and/or $\beta$ are allowed to be periodic.

\begin{Example} \label{todaf} For the matrix $B$ in Example~\ref{b}, $\operatorname{rank}B=6$, and the vector ${\bf v}_1=(1,-2,1,0,0$, $0,0,0)^T$, together with ${\bf v}_2,\ldots,{\bf v}_6$ obtained by shifting the non-zero block to the right, provides a palindromic basis for $\operatorname{im}B$, so the reduction~(\ref{basis}) gives
\begin{gather}\label{udef}
u_m=\frac{\tau_m\tau_{m+2}}{\tau_{m+1}^2}.
\end{gather}
Then the U-system corresponding to (\ref{874}) is the non-autonomous recurrence
\begin{gather}\label{utoda}
u_{m+6}u_m = \frac{\al_m u_{m+5}u_{m+4}^2u_{m+3}^3u_{m+2}^2u_{m+1}+\beta}{u_{m+5}^2u_{m+4}^3u_{m+3}^4u_{m+2}^3u_{m+1}^2}, \qquad \al_{m+4}=\al_m.
\end{gather}
Upon applying Theorem \ref{FH}, we see that (\ref{utoda}) preserves the symplectic form
\begin{gather*}
\hat\omega = \sum_{i<j}\frac{\hat{b}_{ij}}{u_i u_j} \rd u_i \wedge \rd u_j,
\qquad \hat{B}=\big(\hat{b}_{ij}\big)=\left(
\begin{matrix}
0 & 1 & 2 & 3 & 2 & 1 \\
-1 & 0 & 2 & 4 & 4 & 2 \\
-2 & -2 & 0 & 3 & 4 & 3 \\
-3 & -4 & -3 & 0 & 2 & 2 \\
-2 & -4 & -4 & -2 & 0 & 1 \\
-1 & -2 & -3 & -2 & -1 & 0
\end{matrix} \right).
\end{gather*}
Equivalently, each iteration of the corresponding six-dimensional map~$\hat\varphi_m$, which depends on $m\bmod 4$, preserves the nondegenerate log-canonical Poisson bracket (obtained by inverting the matrix~$\hat B$) given by
\begin{gather}\label{tpb}
\{u_i,u_j \} =c_{j-i}u_iu_j ,    \qquad 0 \leq i<j \leq 5,\end{gather}
with $c_1=1$, $c_2=c_5=0$,   $c_3=-c_4=-2$.

Upon comparing (\ref{udef}) with (\ref{zdef}), we obtain the formula
\begin{gather}\label{zform}
z_m =u_mu_{m+1}u_{m+2}u_{m+3}
\end{gather} for all $m$. We can use this to write $z_j$ for $0\leq j\leq 2$ in terms of a set of initial coordinates $u_0,u_1,\ldots,u_5$ for the U-system~(\ref{utoda}). However, for higher shifts of~$z_j$ we note the identity
\begin{gather*}
z_{m+3}z_m=\alpha_mu_{m+3}+\frac{\beta}{z_{m+1}z_{m+2}},
\end{gather*} which follows from (\ref{zform}) and (\ref{utoda}), so after substituting for $z_3$, $z_4$, $z_5$, $z_6$ in (\ref{tolax}) it is clear that, in addition to $u_j$ for $0\leq j\leq 5$, the parameters $\al_0$, $\al_1$, $\al_2$, $\al_3$ will also appear in ${\bf L}={\bf L}_0$ and ${\bf M}={\bf M}_0$ (and similarly for ${\bf L}_m$, ${\bf M}_m$, replacing the index $j$ on each variable in ${\bf L}$, ${\bf M}$ by $m+j$). In this way we obtain the Lax representation for the U-system itself, while from the spectral curve (\ref{todaspec}), the coef\/f\/icients $H_1$, $H_2$, $H_3$ in (\ref{h1}) can be pulled back by the formula~(\ref{zform}), to provide 3 independent functions of $u_j$, in involution with respect to the bracket~(\ref{tpb}) in 6~dimensions (see~\cite{Chloe} for explicit formulae in the case $\al_m=\al= {\rm const}$). Pulling back the fourth coef\/f\/icient we f\/ind
\begin{gather}\label{h4form}
H_4=\al_0\al_1\al_2\al_3,
\end{gather}
which (like any cyclically symmetric function of $\al_0$, $\al_1$, $\al_2$, $\al_3$) is a trivial f\/irst integral of the U-system. (By an abuse of notation, in (\ref{h4form}) and elsewhere we use the same symbol to denote a function and its pullback.) Thus we see that in the autonomous case we have a~6-dimensional symplectic map $\hat\varphi$ given by
\begin{gather*} 
\hat\varphi(u_0,u_1,u_2,u_3,u_4,u_5)=\left(u_1, u_2,u_3,u_4,u_5,\frac{\beta+\alpha u_1 u_2^2 u_3^3 u_4^2 u_5}{u_{0}u_1^2 u_2^3 u_3^4 u_4^3 u_5^2}\right),
\end{gather*}
obtained from (\ref{utoda}) by f\/ixing $\al_m=\al= {\rm const}$, with 3 commuting f\/irst integrals, so this is an integrable system in the Liouville sense. In the non-autonomous case, we have instead a family of symplectic maps $\hat\varphi_m$ cycling with $m\bmod 4$, but we can interpret the fourfold composition $\hat\varphi_3\circ\hat\varphi_2\circ\hat\varphi_1\circ\hat\varphi_0$ as an autonomous system with the three commuting invariants $H_1$, $H_2$, $H_3$ described above, so again this is a Liouville integrable system.

These results also lead to an interpretation of the 7-dimensional map def\/ined by (\ref{co2}) as an integrable system. Indeed, the expression~(\ref{zform}) means that the bracket (\ref{tpb}) can be lifted to the coordinates~$z_j$, $0\leq j\leq 6$, to give a Poisson bracket of rank 6 specif\/ied by
\begin{gather*}
\{z_0,z_1\}=0=\{z_0,z_2\}, \qquad \{z_0,z_3\}=-z_0z_3+\beta (z_1z_2)^{-1}, \\ \{z_0,z_4\}=z_0z_4-\beta^2 \big(z_1^2z_2^2z_3^2\big)^{-1},
\\
\{z_0,z_5\}  = -\beta^2 \big(z_1z_2^2z_3^2z_4\big)^{-1} + \beta^3 \big(z_1^2z_2^3z_3^3z_4^2\big)^{-1}, \\
\{z_0,z_6\}  = -\beta^2 \big(z_1z_2z_3^2z_4z_5\big)^{-1} + \beta^3 \big(z_1^2z_2^2z_3^3z_4^2z_5^2\big)^{-1}(z_1+z_5) -\beta^4 \big(z_1^2z_2^3z_3^4z_4^3z_5^2\big)^{-1};
\end{gather*}
all other brackets follow by shifting indices. The preceding results imply that, as functions of~$z_j$, the quantities~$H_1$,~$H_2$,~$H_3$ in~(\ref{h1}) are in involution with respect to the above bracket, while from the expression~(\ref{h4form}) it follows that~$H_4$ is a Casimir for this bracket; this can also be verif\/ied directly from~(\ref{h1}).
\end{Example}

\begin{Remark} In the previous example, the formula~(\ref{tolax}) is a ``big Lax pair'' for the system: the spectral curve $\Gamma$ in (\ref{todaspec}) has genus~9, while the Liouville tori (level sets of f\/irst integrals) are only 3-dimensional. One way to understand this is by noting that~$\Gamma$ is invariant under the action of~$C_3$ (the cyclic group of order~3), generated by $(\zeta,\xi)\to (e^{2\pi\ri /3}\zeta, e^{2\pi\ri /3}\xi)$, so that $\Gamma$ is a~threefold cover of the curve $\tilde{\Gamma}= \Gamma/C_3$, ramif\/ied at $(0,0)$ and $(\infty,\infty)$, and $\tilde{\Gamma}$ has genus~3. A~more direct way to obtain the curve $\tilde{\Gamma}$ is to note that the system is one of a family of reductions of a lattice equation of discrete Toda type, considered in the next section (see~\eqref{Todaeq} below). This reduction procedure yields $2\times 2$ Lax pairs, which will be described in future work~\cite{DTKQg}.
\end{Remark}

\section{Reductions of discrete Toda type}\label{section3}

In this section we brief\/ly consider the two-parameter family of plane wave reductions of discrete KP with $(\delta_1,\delta_2,\delta_3)=(P,P-Q,0)$, for integers $P>Q$, which includes Example~\ref{todaex} when $P=4$, $Q=1$. These examples arise from travelling wave solutions of the f\/ive-point lattice equation
\begin{gather} \label{Todaeq}
\frac{V_{k,l}}{V_{k+1,l}}-\frac{V_{k-1,l}}{V_{k,l}}+\alpha \left(\frac{V_{k+1,l-1}}{V_{k,l}}-\frac{V_{k,l}}{V_{k-1,l+1}}\right)=0,
\end{gather}
which for $\alpha=p q$ appears in \cite{book}, and can be considered as a discrete time Toda equation \cite{Date}. If we impose the periodicity
\begin{gather*}V_{k+Q,l-P}=V_{k,l}, \end{gather*}
which is called the $(Q,-P)$ periodic reduction of equation \eqref{Todaeq}, then we may write
\begin{gather*}V_{k,l}=v_m,\qquad \mathrm{where} \quad m=k P+l Q \end{gather*}
is a travelling wave variable, and $v_m$ satisf\/ies the ordinary dif\/ference equation
\begin{gather} \label{todared}
\frac{v_m}{v_{m+P}}-\frac{v_{m-P}}{v_m}+\alpha \left(\frac{v_{m+P-Q}}{v_{m}}-\frac{v_m}{v_{m+Q-P}}\right)=0.
\end{gather}

\begin{Proposition} \label{prTodalift} Suppose that
\begin{gather*}v_m=\frac{\tau_m}{\tau_{m+Q}}\end{gather*} is a solution of~\eqref{todared}. Then $\tau_m$ is a solution of
\begin{gather} \label{todaKP}
\tau_{m+2P}\tau_{m}=\alpha \tau_{m+2P-Q}\tau_{m+Q}+ \beta_m \tau_{m+P}^2, \qquad \beta_m=\beta_{m+Q},
\end{gather}
and the converse is also true.
\end{Proposition}
\begin{proof}
By setting $v_m=\frac{\tau_m}{\tau_{m+Q}}$ in \eqref{todared}, up to an overall shift we f\/ind that
\begin{gather*} 
\frac{\tau_{m+2P}\tau_{m}-\alpha \tau_{m+2P-Q}\tau_{m+Q}}{\tau_{m+P}^2}= \frac{\tau_{m+2P+Q}\tau_{m+Q}-\alpha  \tau_{m+2P}\tau_{m+2Q}} {\tau_{m+P+Q}^2},
\end{gather*}
and denoting the left-hand side above by $\beta_m$ we see that this quantity is periodic with period $Q$, which yields the bilinear equation in~(\ref{todaKP}). Conversely, any solution of~(\ref{todaKP}) with a~$Q$-periodic coef\/f\/icient $\beta_m$ provides a solution of~\eqref{todared}.
\end{proof}

\begin{Example} In the running example above, with $(\delta_1,\delta_2,\delta_3)=(4,3,0)$, we have
\begin{gather*}
v_m = \frac{\tau_m}{\tau_{m+1}}, \qquad \mathrm{with} \quad u_m=\frac{\tau_m \tau_{m+2}}{\tau_{m+1}^2}=\frac{v_m}{v_{m+1}} \end{gather*}
being a solution of the autonomous version of the U-system~(\ref{utoda}), where $\al_m=\al=\mathrm{const}$. After shifting once to eliminate $\beta$ from the U-system, an equation of order 7 for $u_m$ arises, and this is equivalent to a relation of order~8 for~$v_m$, namely
 \begin{gather*}\frac{v_{m+4}}{v_{m+8}}-\frac{v_{m}}{v_{m+4}}+\alpha \left(\frac{v_{m+7}}{v_{m+4}}-\frac{v_{m+4}}{v_{m+1}}\right)=0, \end{gather*}
which is the $(1,-4)$ periodic reduction of the discrete Toda equation~\eqref{Todaeq}.
\end{Example}

In the general $(Q,-P)$ periodic reduction of \eqref{Todaeq}, one can consider a Poisson structure for the variables~$v_m$ in~\eqref{todared}, and a ``small'' ($2\times 2$) Lax representation for this reduction and the associated U-system. We propose to treat these details elsewhere~\cite{DTKQg}.

\section{Reductions of discrete KdV type}\label{section4}

In this section, for a pair of integers $L>M$, we consider two dif\/ferent plane wave reductions of discrete KP, corresponding to the choices
\begin{gather*}
(\delta_1,\delta_2,\delta_3) = \left( L + \frac{M}{2}, \frac{M}{2}, L-\frac{M}{2}\right) \qquad \mathrm{or}\qquad
 \left( M + \frac{L}{2}, \frac{L}{2}, \left|M-\frac{L}{2}\right|\right).
\end{gather*}
It turns out that these two dif\/ferent reductions are very closely related to each other: they both correspond to the $(L,M)$ periodic reduction of the lattice KdV equation~\cite{Hirota}
\begin{gather} \label{KdV}
V_{k+1,l}-V_{k,l+1}=\alpha\left(\frac{1}{V_{k,l}}-\frac{1}{V_{k+1,l+1}}\right).
\end{gather}
This connection leads to an alternative $2\times 2$ Lax pair for these discrete KP reductions and their associated U-systems, as well as linking the Liouville integrability of the latter with that of the corresponding discrete KdV reduction.

Observe that if the
 $(L,M)$ reduction is imposed on (\ref{KdV}), then we have
\begin{gather*}V_{k+L,l+M}=V_{k,l}\implies V_{k,l}=v_m, \qquad m=l L-k M, \end{gather*}
and we can write the following ordinary dif\/ference equation in terms of the travelling wave variable $m$:
\begin{gather} \label{RedkdV}
v_{m+{L}+{M}}-v_m=\alpha\left(\frac{1}{v_{m+L}}-\frac{1}{v_{m+M}}
\right).
\end{gather}
\begin{Remark}Up to sending $v_m\to 1/v_m$ and redef\/ining $\al$, the $(L,M)$ reduction and the $(L,-M)$ reduction of \eqref{KdV} are equivalent. The parameter $\al$ can be removed by scaling.
\end{Remark}

\begin{Proposition} \label{kdvred}
Suppose that
\begin{gather}\label{vtau}
v_m=\frac{\tau_m \tau_{m+L+M}}{\tau_{m+M} \tau_{m+L}}\end{gather}
is a solution of \eqref{RedkdV}. Then $\tau_m$ satisfies the following two bilinear equations
\begin{gather}
\tau_{m+2L+M}\tau_m = \beta_m \tau_{m+L+M} \tau_{m+L} - \alpha \tau_{m+2 L}\tau_{m+M}, \qquad \beta_{m+M}=\beta_m, \label{1HKdV} \\
\tau_{m+2M+L}\tau_{m} = \beta'_m \tau_{m+L+M} \tau_{m+M}+\alpha \tau_{m+2 M}\tau_{m+L},\qquad \beta'_{m+L}=\beta'_m. \label{2HKdV}
\end{gather}
Conversely, if $\tau_m$ is a solution of either~\eqref{1HKdV} or~\eqref{2HKdV}, then $v_m$ given by~\eqref{vtau} satisfies~\eqref{RedkdV}.
\end{Proposition}

\begin{proof} With~(\ref{vtau}),~\eqref{RedkdV} is equivalent to either of the two equalities
\begin{gather*}
 \frac{\tau_{m+2L+M}\tau_m+\alpha \tau_{m+2 L}\tau_{m+M} }{\tau_{m+L}\tau_{m+L+M}} = \frac{\tau_{m+2L+2M}\tau_{m+M}+\alpha \tau_{m+2L+M}\tau_{m+2M} }{\tau_{m+L+M}\tau_{m+L+2M}},
\\
\frac{\tau_{m+2M+L}\tau_m-\alpha \tau_{m+2 M}\tau_{m+L} }{\tau_{m+M}\tau_{m+M+L}}= \frac{\tau_{m+2M+2L}\tau_{m+L}-\alpha \tau_{m+2 M+L}\tau_{m+2L} }{\tau_{m+M+L}\tau_{m+M+2L}},
\end{gather*}
from which the result follows.
\end{proof}

\begin{Remark}
In \cite[Proposition 8]{HKQT}, it was proved that the solutions of $(d - 1, -1)$ reductions of the lattice KdV equation \eqref{KdV} are Liouville integrable, using the observation that (with $\al=-1$) these reductions are given in terms of a tau function that satisf\/ies the bilinear recurrence relation
\begin{gather*}\tau_{m+d+1} \tau_m = \beta_m \tau_{m+d} \tau_{m+1} + \tau_{m+d-1} \tau_{m+2},\end{gather*}
with the coef\/f\/icient $\beta_{m}$ having period $d-1$. The above result extends this observation to the general $(L,M)$ reduction, and shows that in each case there are actually two dif\/ferent bilinear equations involved.
\end{Remark}

\subsection[$2 \times 2$ Lax pairs]{$\boldsymbol{2 \times 2}$ Lax pairs}
Without loss of generality, we can assume from now on that~$L$, $M$ are coprime (since otherwise the equations split into copies of systems in lower
dimension). According to Corollary~\ref{LaxPair}, the reductions~\eqref{1HKdV} and~\eqref{2HKdV} of discrete KP each admit a Lax representation,
with $L \times L$ and $\min{(L,2M)} \times \min{(L,2M)}$ Lax matrices respectively. However, in these cases there is also a~$2\times2$ Lax representation, derived from the Lax representation of the lattice KdV equation.

The lattice KdV equation \eqref{KdV} admits a zero curvature representation with a~$2\times 2$ Lax pair. Specif\/ically, equation \eqref{KdV} is equivalent to
\begin{gather} \label{LaxKdV}
{\bf L}(V_{k,l+1},V_{k+1,l+1},\la){\bf M}(V_{k,l},\la)={\bf M}(V_{k+1,l},\la){\bf L}(V_{k,l},V_{k+1,l},\la),
\end{gather}
where $\la$ is a spectral parameter and
\begin{gather} \label{LaxPair2}
{\bf L}(V,W,\la)=\left(\begin{matrix} V-\frac{\alpha}{W} & \la \\ 1 & 0 \end{matrix}\right), \qquad {\bf M} (V,\la)=\left( \begin{matrix} V & \la \\  1 & \frac{\al}{V} \end{matrix}\right).
\end{gather}

It is well known that the Lax representation of quadrilateral lattice equations gives rise to Lax representations of their periodic reductions (see, e.g., \cite{staircase} and references). Hence, a $2 \times 2$ Lax representation can be obtained for the~$(L,M)$ periodic reduction~\eqref{RedkdV} and consequently for the corresponding discrete KP reductions, as well as their associated U-systems. First integrals of these systems are derived from the spectrum of their corresponding monodromy matrix.

\subsection{Example: a discrete KdV reduction of order 5}

The Liouville integrability of \eqref{RedkdV} in the case $L=4$, $M=1$ follows from the results of~\cite{HKQT}, so here we consider a dif\/ferent example of order 5, namely the case $L=3$, $M=2$. In the latter case, the recurrences~(\ref{1HKdV}) and~(\ref{2HKdV}) become
\begin{gather}
\tau_{m+8}\tau_{m} = -\alpha \tau_{m+6}\tau_{m+2} + \beta_m \tau_{m+3}\tau_{m+5}, \qquad \beta_{m+2}=\beta_m, \label{32m1} \\
\tau_{m+7}\tau_m = \alpha \tau_{m+4}\tau_{m+3} +\beta'_m \tau_{m+2}\tau_{m+5}, \qquad \beta'_{m+3}=\beta'_m, \label{32m2}
\end{gather}
respectively.

In each case, the associated exchange matrix $B$ has rank~4. The corresponding U-systems are obtained by setting
\begin{gather}\label{usystau}
u_m=\frac{\tau_{m} \tau_{m+4}}{\tau_{m+1} \tau_{m+3}}, \qquad u'_m=\frac{\tau_{m} \tau_{m+3}}{\tau_{m+1} \tau_{m+2}},
\end{gather}
to get
\begin{gather}
u_m u_{m+1}u_{m+2}u_{m+3}u_{m+4} = \beta_m-\al u_{m+2} , \label{32um1} \\
u_m' u_{m+1}'({u}_{m+2}')^2 u_{m+3}'u_{m+4}' = \beta'_m u_{m+2}'+\al, \label{32um2}
\end{gather}
respectively. By Theorem~\cite{FH}, the latter are equivalent to iteration of 4-dimensional birational symplectic maps $\hat\varphi_m^{(j)}$, $j=1,2$, where
\begin{gather*}
\hat\varphi_0^{(1)}(u_0,u_1,u_2,u_3) = \left(u_1,u_2,u_3,\frac{\beta_0-\alpha u_2 }{u_0 u_1 u_2 u_3}\right), \\
\hat\varphi_0^{(2)} (u_0',u_1',u_2',u_3')= \left(u_1',u_2',u_3',\frac{\alpha+\beta'_0u_2' }{u_0' u_1' (u_2')^2 u_3'}\right) ,
\end{gather*}
respectively (we have just written the case $m=0$), with the nondegenerate Poisson bracket being specif\/ied by
 \begin{gather}\label{br1}
\{u_j,u_{j+1}\}_1 =0, \qquad \{u_j,u_{j+2}\}_1 =u_ju_{j+2}, \qquad \{u_j,u_{j+3}\}_1 =-u_ju_{j+3}
\end{gather} for the f\/irst one and{\samepage
 \begin{gather}\label{br2}
\{u_j',u_{j+1}'\}_2 =0, \qquad \{u_j',u_{j+2}'\}_2 =u_j'u_{j+2}', \qquad \{u_j',u_{j+3}'\}_2 =-u_j'u_{j+3}'
\end{gather}
for the second. Observe that the two brackets~(\ref{br1}) and~(\ref{br2}) are identical.}

On the other hand, from Proposition \ref{kdvred}, by setting
\begin{gather}\label{vtaus} v_j=\frac{\tau_j \tau_{j+5}}{\tau_{j+2} \tau_{j+3}},\end{gather}
the discrete KP reductions (\ref{32m1}) and (\ref{32m2}) both yield the $(3,2)$ periodic reduction of the lattice KdV equation, which is equivalent to the 5-dimensional birational map
\begin{gather} \label{32KdV}
(v_0,v_1,v_2,v_3,v_4)\mapsto \left(v_1,v_2,v_3,v_4,v_0+\alpha\left(\frac{1}{v_{3}}-\frac{1}{v_{2}}\right)\right).
\end{gather}

{\bf Lax representation and f\/irst integrals.} Corollary~\ref{LaxPair} produces a $3\times 3$ Lax representation for both (\ref{32m1}) and (\ref{32m2}), and for their corresponding U-systems (\ref{32um1}) and (\ref{32um2}), with a~trigonal spectral curve. However, it is more straightforward to apply the~(3,2) periodic reduction to the discrete KdV Lax pair~\eqref{LaxKdV}, directly giving a Lax pair for~(\ref{32KdV}) in terms of the coordinates~$v_j$, which can then be rewritten in terms of the~$u_j$ or~$u_j'$ as desired.

The monodromy matrix of the $(3,2)$ KdV periodic reduction, obtained by the staircase method \cite{staircase}, is
\begin{gather*}\mathcal{M}(v_0,v_1,v_2,v_3,v_4; \la)={\bf M}(v_3){\bf L}(v_1,v_3){\bf M}(v_4){\bf L} (v_2,v_4){\bf L} (v_0,v_2),\end{gather*}
where the $2 \times 2$ matrices ${\bf M}$ and ${\bf L}$ are given in \eqref{LaxPair2}. The map \eqref{32KdV} satisf\/ies the equation
\begin{gather*}\mathcal{M}(v_0,v_1,v_2,v_3,v_4) \mathcal{L}=\mathcal{L}\mathcal{M}(v_1,v_2,v_3,v_4,v_5),\end{gather*}
where
\begin{gather*}\mathcal{L}={\bf L}(v_0,v_2)^{-1}{\bf M}(v_5){\bf L}(v_3,v_5){\bf L}(v_1,v_3) \qquad \text{and} \qquad v_5=v_0+\alpha\left(\frac{1}{v_{3}}-\frac{1}{v_{2}}\right).\end{gather*}
The associated hyperelliptic spectral curve in the $(\la,\mu)$ plane is of genus~2, being given by
\begin{gather*}
\det ({\cal M}(\la)-\mu \mathbf{I})\equiv \mu^2-\Pi(\la)\mu -\la^3(\la -\al )^2=0,
\end{gather*}
where the trace of the monodromy matrix ${\cal M}(\la)=\mathcal{M}(v_0,v_1,v_2,v_3,v_4;\la)$ has the form
\begin{gather*}
\Pi (\la)=H_2\la^2+H_1\la +H_0.
\end{gather*}
From the trace of the monodromy matrix we f\/ind three functionally independent f\/irst integrals for the map \eqref{32KdV}, which are conveniently chosen as
\begin{gather*}
I_1=\Pi(0)=H_0, \qquad I_2=H_2, \qquad I_3=\Pi(\al)=H_2\al^2+H_1\al+H_0,
\end{gather*}
so that they have the explicit form
\begin{gather*}
I_1 = -\frac{1}{v_2}(\alpha-v_0 v_2)(\alpha-v_1 v_3)(\alpha-v_2 v_4), \\
I_2 = v_0+v_1+v_2+v_3+v_4-\frac{\alpha}{v_2}, \\
I_3 = v_2(\alpha+v_0 v_3)(\alpha+v_1 v_4).
\end{gather*}

In order to obtain the corresponding f\/irst integrals in terms of the variables for the two dif\/ferent U-systems, we compare~(\ref{usystau}) with~(\ref{vtaus}) to see that we can write the variables for~\eqref{32KdV} in two dif\/ferent ways, as
\begin{gather*}
v_m=u_m u_{m+1}=u_m' u_{m+1}'u_{m+2}'.
\end{gather*}
Now to pull back $I_1$, $I_2$, $I_3$ to the f\/irst U-system, we must iterate the recurrence \eqref{32um1} to get
\begin{gather} \label{vu1}
v_0=u_0 u_1, \qquad v_1=u_1 u_2, \qquad v_2=u_2 u_3, \qquad v_3= \frac{\beta_0-\alpha u_2}{u_0 u_1 u_2}, \qquad v_4=\frac{\beta_1-\alpha u_3}{u_1 u_2 u_3}.
\end{gather}
In that case, we f\/ind that $I_1$, $I_2$ are two independent functions of the $u_j$, while pulling back the third integral yields
\begin{gather}\label{cas1}
I_3=\beta_0 \beta_1,
\end{gather}
which is a trivial f\/irst integral for the symplectic map $\hat\varphi_m^{(1)}$. Similarly, by iterating the second U-system \eqref{32um1}, we obtain
\begin{gather} \label{vu2}
 v_0 = u_0' u_1' u_2', \qquad v_1 = u_1' u_2' u_3', \qquad  v_2 = \frac{\alpha+\beta'_0 u_2'}{u_0' u_1' u_2'}, \\
 v_3 = \frac{\alpha+\beta'_1 u_3'}{u_1' u_2' u_3'}, \qquad
v_4 = \frac{\alpha \beta'_2+\beta'_0 \beta'_2 u_2'+\al u_0' u_1' (u_2')^2 u_3'}{\al u_2' u_3' +\beta_0' (u_2')^2 u_3'}. \nonumber
\end{gather}
The quantities $I_2$, $I_3$ pull back to two independent functions of the~$u_j'$, while for the f\/irst quantity we f\/ind a trivial f\/irst
integral of the symplectic map~$\hat\varphi_m^{(2)}$, namely
\begin{gather}\label{cas2}
I_1=\beta_0' \beta_1'\beta_2'.
\end{gather}

{\bf Bi-Hamiltonian structure and Liouville integrability.} From (the pullbacks of) the formulae for $I_1$, $I_2$ one can verify directly that, as functions of $u_j$, they are in involution with respect to the bracket~\eqref{br1}, which implies that the symplectic map $\hat\varphi_m^{(1)}$ is Liouville integrable. Similarly, one can check that the same conclusion holds for $\hat\varphi_m^{(2)}$, by using~(\ref{br2}) to verify that $\{I_2,I_3\}_2=0$. However, there is another way to obtain this result, by lifting the brackets for both U-systems to obtain two dif\/ferent Poisson structures for~\eqref{32KdV}.

From the f\/irst U-system, using the formulae (\ref{vu1}) we f\/ind that the Poisson bracket $\{\,,\,\}_1$ in~\eqref{br1} pushes forward to a~bracket in 5 dimensions (denoted here by the same symbol):
\begin{alignat}{3}
&\{v_0,v_1\}_1=v_0v_1, \qquad &&  \{v_0,v_2\}_1=v_0v_2,&\nonumber \\
& \{v_0,v_3\}_1=-v_0v_3-\alpha, \qquad && \{v_0,v_4\}_1=-v_0v_4. &\label{brv1}
\end{alignat}
The quantities $I_1$, $I_2$, $I_3$ found from the trace of the monodromy matrix are in involution with respect to the bracket~(\ref{brv1}). This Poisson bracket has rank~4, with~$I_3$ as a Casimir (this follows from the expression~(\ref{cas1}): a constant function of the~$u_j$ must lift to a~Casimir).

Similarly, pushing forward the second U-system, using the formulae (\ref{vu2}) we f\/ind that the Poisson bracket in~\eqref{br2} lifts to
\begin{alignat}{3}
& \{v_0,v_1\}_2=v_0v_1, \qquad && \{v_0,v_2\}_2=v_0v_2-\alpha, & \\
& \{v_0,v_3\}_2=-v_0v_3, \qquad && \{v_0,v_4\}_2=-v_0v_4+\frac{\alpha^2}{v_2^2}.& \label{brv2}
\end{alignat}
Once again, this is a bracket of rank 4, with~$I_1$ as a Casimir (as follows from~(\ref{cas2}) above). One can show directly that~$I_2$, $I_3$ provide two more independent commuting functions of~$v_j$ with respect to the bracket $\{\, , \, \}_2$ given by~(\ref{brv2}).

It turns out that these two brackets for~\eqref{32KdV} are compatible, in the sense that their sum (and hence any linear combination) also satisf\/ies the Jacobi identity, so is also a Poisson bracket. In fact, we observe that the dif\/ference of the Poisson brackets
\begin{gather}\label{dif}
\{\cdot,\cdot \}_3=\{\cdot,\cdot \}_1-\{\cdot,\cdot \}_2
\end{gather}
coincides (under the transformation $v_i \mapsto \frac{1}{v_i}$ and by inserting the parameter $\alpha$) with the one that is derived from the Lagrangian structure of the lattice KdV equation, recently presented in~\cite{Dinh}. Thus we see that the map \eqref{32KdV} has a~bi-Hamiltonian structure, and the sequence of f\/irst integrals $I_1$, $I_2$, $I_3$ gives a f\/inite Lenard--Magri chain.

\subsection{Comments on the integrability of the general case}

 In the case $M=1$, the Liouville integrability of $(L,1)$ periodic reductions of discrete KdV was proved in~\cite{HKQT}. However, the Liouville integrability for the case of general $(L,M)$ will be the subject of a future publication, and we only comment on it brief\/ly here. The Liouville integrability of the corresponding U-systems follows from the Liouville integrability of the $(L,M)$ periodic reduction of the lattice KdV equation (and vice-versa). For all $(L,M)$ we f\/ind that the two U-systems are of the same dimension and preserve the same log-canonical symplectic structure.

For $L+M$ odd, we can always f\/ind two compatible Poisson structures $\{\, ,\,\}_{1,2}$ for the variab\-les~$v_m$, whose dif\/ference~(\ref{dif}) coincides with the bracket obtained from a discrete Lagrangian in~\cite{Dinh}. Here we just state the corresponding theorem.
\begin{Theorem}
Let $L$, $M$ be coprime with $L>M>1$ and $L+M$ odd. For $0 \leq i<j \leq L+M-1$, the brackets
\begin{gather*}
\{ v_i, v_j \}_1 = \begin{cases}
c_{j-i} v_i v_j, &j-i \neq L, \\
c_{j-i} v_i v_j+ c_N \alpha, &j-i = L,
\end{cases} \\
\{ v_i, v_j \}_2 = \begin{cases}
c_{j-i} v_i v_j, &j-i \neq kM, \\
\displaystyle c_{j-i} v_i v_j +c_M (-\alpha)^k \prod\limits_{l=1}^{k-1} v_{i+l\, M}^{-2}, &j-i = k M,
\end{cases}
\end{gather*}
where
\begin{gather*} 
c_k=(-1)^h , \qquad \text{with} \quad h=\frac{k}{M} \bmod(N+M),
\end{gather*}
for $k=1,\dots, L+M-1$, define two compatible Poisson structures on $\mathbb{C}^{L+M}$ of rank $L+M-1$ preserved by the map
\begin{gather*} 
(v_0,v_1,\dots,v_{L+M-1}) \mapsto \left(v_1,v_2,\dots,v_0+\alpha\left(\frac{1}{v_{L}}-\frac{1}{v_{M}}\right)\right),
\end{gather*}
corresponding to the $(L,M)$ periodic reduction of the lattice KdV equation.
\end{Theorem}

The proof of this theorem, as well as the detailed description of the U-systems in the general case, will be presented elsewhere.

In the case where $L+M$ is even the situation is slightly dif\/ferent. In this case the U-systems are $(L+M-2)$-dimensional maps but can be lifted to a space of one dimension higher. In this $(L+M-1)$-dimensional space there is an invariant bi-Poisson structure that ensures Liouville integrability. The integrability of the lifted maps implies the integrability of the initial U-systems as well as that of the corresponding $(L,M)$ reductions of the discrete KdV equation.

\section{Conclusions}\label{section5}
We have described the general properties of plane wave reductions of the Hirota--Miwa (discrete KP) equation~\eqref{Hirota}, and have considered two special families of such reductions that correspond to travelling waves of certain lattice equations, of discrete Toda and discrete KdV type, respectively. Both of these families admit a $2\times 2$ Lax representation that gives rise to hyperelliptic spectral curves, whose coef\/f\/icients are f\/irst integrals of the corresponding maps. This is in addition to the Lax pairs obtained in Corollary~\ref{LaxPair}, which are generally of larger size.

The plane wave reductions of the discrete KP equation yield f\/inite genus solutions, and in that case the equation \eqref{Hirota} itself corresponds to the Fay trisecant identity for the theta function of the spectral curve~\cite{finitegenus}. The Fay identity can also be used to derive corresponding solutions of continuous soliton equations via a limiting process~\cite{mumford,taimanov}, so in a sense the discrete Hirota equation \eqref{Hirota} is more fundamental than its continuous counterparts.

It is interesting to note that the discrete Toda and discrete KdV families exhaust all the discrete KP reductions up to order~7, i.e.\ all the three-term Somos recurrences up to Somos-7. If we proceed to higher order discrete KP reductions, then new families appear. For example, the three-term Somos-8 recurrence
\begin{gather*}
\tau_{m+8}\tau_m=\al \tau_{m+7}\tau_{m+1}+\be \tau_{m+5} \tau_{m+3} \end{gather*}
is neither of Toda nor of KdV type: it belongs to a dif\/ferent family of recurrences associated with periodic reductions of a Boussinesq type lattice equation. Concerning three-term Somos-9 and Somos-10 recurrences, all cases except one are included in the Toda, KdV or Boussinesq families. Further details of the Liouville integrable maps arising from these and the other families will be the subject of future work.

\subsection*{Acknowledgements}
Some of these results f\/irst appeared in the Ph.D.~Thesis~\cite{Chloe}, which was supported by EPSRC studentship EP/P50421X/1. ANWH is supported by EPSRC fellowship EP/M004333/1.

\pdfbookmark[1]{References}{ref}
\LastPageEnding

\end{document}